\documentclass[twoside]{article}

\usepackage[accepted]{aistats2020}

\usepackage[round]{natbib}
\usepackage{graphicx}
\usepackage[utf8]{inputenc} 
\usepackage[T1]{fontenc}    
\usepackage{hyperref}       
\usepackage{url}            
\usepackage{booktabs}       
\usepackage{amsfonts}       
\usepackage{nicefrac}       
\usepackage{graphicx}
\usepackage{subcaption}
\usepackage{chngcntr}

\usepackage{xargs}                      
\usepackage[pdftex,dvipsnames]{xcolor}  
\usepackage[colorinlistoftodos,prependcaption,textsize=tiny]{todonotes}

\usepackage{amsmath, amsthm, amssymb, bbm}
\usepackage{color}
\usepackage{algorithmic,algorithm}

\newtheorem{theorem}{Theorem}[section]

\newtheorem{lemma}[theorem]{Lemma}
\newtheorem{condition}[theorem]{Condition}

\theoremstyle{remark}
\newtheorem{remark}{Remark}[section]

\def\M2{MTP$_2$}
\def\bR{\mathbf{R}}
\def\R{\mathbb{R}}
\def\bP{\mathbf{P}}
\def\E{\mathcal{E}}
\def\TP{\text{TP}}
\def\TN{\text{TN}}
\def\FP{\text{FP}}
\def\FN{\text{FN}}
\newcommand\indep{\protect\mathpalette{\protect\independenT}{\perp}}

\newcommand{\brh}{\boldsymbol{\rho}}
\newcommand{\bL}{\mathbf{L}}

\newcommandx{\addition}[2][1=]{\todo[linecolor=Plum,backgroundcolor=Plum!25,bordercolor=Plum,#1]{#2}}
\def\independenT#1#2{\mathrel{\rlap{$#1#2$}\mkern2mu{#1#2}}}

\newcommand{\hatOmega}{\widehat{\Omega}}
\DeclareMathOperator{\tr}{trace}
\def\adj{\textrm{adj}}

\runningtitle{Learning High-dimensional GGMs under \M2 without Adjustment of Tuning Parameters}

\begin{document}

\twocolumn[

\aistatstitle{Learning High-dimensional Gaussian Graphical Models \\ under Total Positivity without Adjustment of Tuning Parameters}

\aistatsauthor{ Yuhao Wang \And Uma Roy \And  Caroline Uhler}

\aistatsaddress{ 
University of Cambridge \\ 
\texttt{yw505@cam.ac.uk} 
\And  
Google Research\\
\texttt{uma.roy.us@gmail.com} 
\And  
Massachusetts Institute of Technology\\
\texttt{cuhler@mit.edu }} ]

\begin{abstract}
We consider the problem of estimating an undirected Gaussian graphical model when the underlying distribution is multivariate totally positive of order 2 (\M2), a strong form of positive dependence. Such distributions are relevant for example for portfolio selection, since assets are usually positively dependent. A large body of methods have been proposed for learning undirected graphical models without the \M2 constraint. A major limitation of these methods is that their structure recovery guarantees in the high-dimensional setting usually require a particular choice of a tuning parameter, which is unknown a priori in real world applications. 
We here propose a new method to estimate the underlying undirected graphical model under \M2 and show that it is provably consistent in structure recovery without adjusting the tuning parameters. This is achieved by a constraint-based estimator that infers the structure of the underlying graphical model by testing the signs of the empirical partial correlation coefficients. We evaluate the performance of our estimator in simulations and on financial data.
\end{abstract}

\section{Introduction}
\label{sec:intro}

Gaining insights into complex phenomena often requires characterizing the relationships among a large number of variables. Gaussian graphical models offer a powerful framework for representing high-dimensional distributions by capturing the conditional dependencies between the variables of interest in the form of a network. These models have been extensively used in a wide variety of domains ranging from speech recognition~\citep{johnson2012mathematical} to genomics~\citep{kishino2000correspondence} and finance~\citep{wang2011dynamic}.

In this paper we consider the problem of learning a Gaussian graphical model under the constraint that the distribution is multivariate totally positive of order 2 (\M2), or equivalently, that all partial correlations are non-negative. Such models are also known as attractive Gaussian random fields. \M2 was first studied in~\citep{B82,FKG71,KarlinRinott80,KR83} and later also in the context of graphical models~\citep{Fallat17,LUZ17}. \M2 is a strong form of positive dependence, which is relevant for modeling in various applications including phylogenetics or portfolio selection, where the shared ancestry or latent global market variable often lead to positive dependence among the observed variables~\citep{MS05,Zwiernik15}.

Due to the explosion of data where the number of variables $p$ is comparable to or larger than the number of samples $N$,  the problem of learning undirected Gaussian graphical models in the high-dimensional setting has been a central topic in machine learning, statistics and optimization. There are two main classes of algorithms for structure estimation for Gaussian graphical models in the high-dimensional setting. A first class of algorithms attempts to explicitly recover which edges exist in the graphical model, for example using conditional independence tests~\citep{Anand12,ST18} or neighborhood selection~\citep{MB06}. A second class of algorithms instead focuses on estimating the precision matrix. 
The most prominent of these algorithms is \emph{graphical lasso}~\citep{BGA08,FHT08,RWRY11,YL07}, which applies an $\ell_1$ penalty to the log-likelihood function to estimate the precision matrix. Other algorithms include linear programming based approaches such as graphical Dantzig~\citep{Yuan10} and CLIME~\citep{CLL11,CLZ16}; optimization with non-convex penalties like~\citep{FFW09,LF09,LW17}; as well as greedy methods like~\citep{JJR12,SPZ12}. 

The main limitation of all aforementioned approaches is the requirement of a specific tuning parameter to obtain consistency guarantees in estimating the edges of the underlying graphical model. In most real-world applications, the correct tuning parameter is unknown and difficult to discover. To make the estimate less sensitive to misspecification of tuning parameters,~\citet{LWang17} and~\citet{SZ13} proposed estimating high-dimensional precision matrices using square-root lasso~\citep{BCW11} and scaled lasso~\citep{SZ12} respectively. These estimators have the advantage that their theoretical guarantees do not rely on an unknown tuning parameter, thereby allowing them to consistently estimate precision matrices without tuning parameter adjustment. While the estimated precision matrices from these methods are guaranteed to converge to the true precision matrix, the zero patterns of the estimated matrices are not guaranteed to recover the underlying graph.

The algorithms described above are for learning the underlying undirected graph in \emph{general} Gaussian models. In this paper, we consider the special setting of \M2 Gaussian models. Several algorithms have been proposed  
that are able to exploit the additional structure imposed by \M2 
with the goal of obtaining stronger results than for general Gaussian graphical models. 
In particular,~\citet{LUZ17} showed that the MLE exists whenever the sample size $N > 2$ (independent of the number of variables $p$), which is striking given that $N > p$ is required for the MLE to exist in general Gaussian graphical models. 
Since the MLE under \M2 is not a consistent estimator for the structure of the graph~\citep{SH15},~\citet{SH15} considered applying thresholding to entries in the MLE, but this procedure requires a tuning parameter and does not have consistency guarantees.


The three main contributions of this paper are: 
\begin{enumerate}
\item[1)] we provide a new algorithm for learning Gaussian graphical models under \M2  that is based on conditional independence testing; 
\item[2)] we prove that this algorithm does not require adjusting any tuning parameters for the theoretical consistency guarantees in structure recovery; 
\item[3)] we show that our algorithm compares favorably to other methods for learning graphical models on both simulated data and financial data.
\end{enumerate}

\section{Preliminaries and Related Work}
\label{sec:pre}

\paragraph{Gaussian graphical models:} Given a graph $G = ([p], \E)$ with vertex set $[p] = \{1, \cdots, p\}$ and edge set $\E$ we associate to each node $i$ in $G$ a random variable $X_i$. A distribution $\bP$ on the nodes $[p]$ forms an \emph{undirected graphical model} with respect to $G$ if 
\begin{equation}
    \label{eq_pair}
X_i\indep X_j\mid X_{[p]\setminus\{i,j\}} \quad \textrm{for all } (i,j)\notin E.
\end{equation}
When $\bP$ is Gaussian with mean zero, covariance matrix $\Sigma$ and precision matrix $\Theta:=\Sigma^{-1}$, the setting we concentrate on in this paper, then (\ref{eq_pair}) is equivalent to  $\Theta_{ij}= 0$ for all $(i,j)\notin E$. By the Hammersley-Clifford Theorem, for strictly positive densities such as the Gaussian, (\ref{eq_pair}) is equivalent to 
\begin{equation*}
X_i\indep X_j\mid X_S \quad \textrm{for all } S\subseteq[p]\setminus\{i,j\} \textrm{ that separate } i,j,
\end{equation*}
where $i,j$ are separated by $S$ in $G$ when $i$ and $j$ are in different connected components of $G$ after removing the nodes $S$ from $G$. 
In the Gaussian setting, $X_i\indep X_j\mid X_S$ if and only if the corresponding \emph{partial correlation coefficient} $\rho_{ij \mid S}$ is zero, which can be calculated from submatrices of $\Sigma$, namely
\begin{align*}
    \rho_{ij \mid S} = & - \frac{((\Sigma_{M,M})^{-1})_{i, j}}{\sqrt{((\Sigma_{M,M})^{-1})_{i, i} ((\Sigma_{M,M})^{-1})_{j, j}}}, \\
     & \textrm{where} M=S\cup\{i,j\}. 
\end{align*}

\paragraph{\M2 distributions:} 
A density function $f$ on $\mathbb{R}^p$ is \M2 if 
$$ f(x) f(y) \leq f(x \wedge y) f(x \vee y) \quad \textrm{for all } x, y \in \mathbb{R}^p,$$
where $\vee, \wedge$ denote the coordinate-wise minimum and maximum respectively~\citep{FKG71,KarlinRinott80}. 
In particular, a Gaussian distribution is \M2 if and only if its precision matrix $\Theta$ is an $M$-matrix, i.e. $\Theta_{ij} \leq 0$ for all $i \neq j$~\citep{B82,KR83}. This implies that all partial correlation coefficients are non-negative, i.e., $\rho_{ij \mid S} \geq 0$ for all $i, j, S$~\citep{KR83}. In addition, for \M2 distributions it holds that $X_i \indep X_j \mid X_S$ if and only if $i,j$ are separated in $G$ given $S$~\citep{Fallat17}. 
Hence $i,j$ are connected in $G$ given $S$ if and only if 
$\rho_{ij \mid S} > 0$.

\M2 distributions are relevant for various applications. In particular, Gaussian tree models with latent variables are \M2 up to sign~\citep{LUZ17}; this includes the important class of single factor analysis models. As an example, in~\citep{SH15} \M2 was used for data measuring students' performance on different math subjects, an application where a factor analysis model with a single latent factor measuring general mathematical ability seems fitting. In addition, factor analysis models are used frequently in psychology and finance; the \M2 constraint has been applied to a dataset from psychology in~\citep{LUZ17} and auctions in~\citep{HLP12}. 
\M2 was also used in the modelling of global stock prices, motivated by the fact that asset price changes are usually positively correlated~\citep{ARU19}; in particular, the authors reported that the correlation matrix of the daily returns of $5$ global stocks is an inverse M-matrix~\citep[Figure~1]{ARU19}. In the same paper, the authors also showed that using a covariance matrix among stocks estimated under \M2 achieves better performance at portfolio selection than other state-of-the-art methods.

\paragraph{Algorithms for learning Gaussian graphical models:} An algorithm is called \emph{consistent} if the estimated graph converges to the true graph $G$ as the sample size $N$ goes to infinity. \emph{CMIT}, an algorithm proposed in~\citep{Anand12}, is most related to the approach in this paper.  Starting in the complete graph, edge $(i,j)$ is removed if there exists $S\subseteq [p]\setminus\{i,j\}$ with $|S| \leq \eta$ (for a tuning parameter $\eta$ that represents the maximum degree of the underlying graph) such that the corresponding empirical partial correlation coefficient satisfies $|\hat{\rho}_{ij \mid S}| \leq \lambda_{N,p}$. For consistent estimation, the tuning parameter $\lambda_{N,p}$ needs to be selected carefully depending on the sample size $N$ and number of nodes $p$. Intuitively, if $(i,j) \notin G$, then $\rho_{ij | S} = 0$ for all $S$ that separate $(i,j)$. Since $\hat{\rho}_{ij \mid S}$ concentrates around $\rho_{ij \mid S}$, it holds with high probability that there exists $S\subseteq [p]\setminus\{i,j\}$ for which $|\hat{\rho}_{ij \mid S}|\leq \lambda_{N,p}$, 
so that edge $(i,j)$ is removed from $G$. Other estimators such as graphical lasso~\citep{RWRY11} and neighborhood selection~\citep{MB06} also require a tuning parameter: $\lambda_{N,p}$ represents the coefficient of the $\ell_1$ penalty and critically depends on $N$ and $p$ for consistent estimation. 
Finally, with respect to estimation specifically under the \M2 constraint, the authors in~\citep{SH15} propose thresholding the MLE $\hatOmega$ of the precision matrix, which can be obtained by solving the following convex optimization problem:
%
\begin{equation} \label{eq:SH_obj}
\hatOmega := \min_{\Omega \succeq 0, \; \Omega_{ij} \leq 0\; \forall i \neq j} - \log \det (\Omega) + \tr (\Omega \hat{\Sigma}),
\end{equation}
%
where $\hat{\Sigma}$ is the sample covariance matrix. The threshold quantile $q$ is a tuning parameter, and apart from empirical evidence that thresholding works well, there are no known theoretical consistency guarantees for this procedure.

In addition to relying on a specific tuning parameter for consistent estimation, existing estimators require additional conditions with respect to the underlying distribution. The consistency guarantees of graphical lasso~\citep{RWRY11} and moment matching approaches such as CLIME~\citep{CLL11} require that the diagonal elements of $\Sigma$ are upper bounded by a constant and that the minimum edge weight
$
\min_{i\neq j, \Theta_{ij} \neq 0} |\Theta_{ij}| \geq C \sqrt{\log (p)/N}
$
for some positive constant $C$. Consistency of CMIT~\citep{Anand12} also requires the minimum edge weight condition. Consistency of CLIME  requires a bounded matrix $L_1$ norm of the precision matrix $\Theta$, which implies that all diagonal elements of $\Theta$ are bounded.

\paragraph{Learning a precision matrix without adjusting any tuning parameters:} Another recent line of work similar to ours considers estimating high-dimensional Gaussian precision matrices without the tuning of parameters. The most prominent such approach is TIGER~\citep{LWang17} and related works include scaled and organic lasso~\citep{SZ12,YB19}. 
These estimators have the desirable property that the estimated precision matrix $\hat{\Theta}$ is guaranteed to converge to the true $\Theta$ without requiring any adjustment of the regularization parameter. However, the  support of the estimated $\hat{\Theta}$ is not guaranteed to converge to the underlying graph $G$ (see e.g. Theorem~4.3 of~\citep{LWang17}), which is the particular task we are interested in this paper.

\section{Algorithm and Consistency Guarantees}
\label{sec:alg}



Algorithm~\ref{alg:mtp2} is our proposed procedure for learning a Gaussian graphical model under the \M2 constraint. In the following, we first describe Algorithm~\ref{alg:mtp2} in detail and then prove its consistency without the need of performing any adjustment of tuning parameters.

\begin{algorithm*}[!t]
	\caption{Structure learning under total positivity}
	\label{alg:mtp2}
	\textbf{Input:} Matrix of observations $\hat{X} \in \bR^{N \times p}$ with sample size $N$ on $p$ nodes. \\
	\textbf{Output:} Estimated graph $\hat{G}$.
	\begin{algorithmic}[1]
		\STATE Set $\hat{G}$ as the completely connected graph over the vertex set $[p]$; set $\ell := -1$;
		\REPEAT
		\STATE set $\ell = \ell + 1$;
		\REPEAT
		\STATE select a (new) ordered pair $(i,j)$ that are adjacent in $\hat{G}$ and such that $|\adj_i(\hat{G}) \setminus \{j\}| \geq \ell$;
		\REPEAT
		\STATE choose a (new) subset $S \subseteq \adj_i(\hat{G}) \setminus \{j\}$ with $|S| = \ell$ and then choose a (new) node $k \in [p] \setminus S \cup \{i,j\}$;
		\STATE calculate the empirical partial coefficient $\hat{\rho}_{ij \mid S \cup \{k\}}$ using randomly drawn data with batch size $M := N^\gamma$; if $\hat{\rho}_{ij \mid S \cup \{k\}} < 0$, delete $i - j$ from $\hat{G}$;
		\UNTIL{edge $i - j$ is deleted from $\hat{G}$ or all $S$ and $k$ are considered;}
		\UNTIL{all ordered pairs $i,j$ that are adjacent in $\hat{G}$ with $|\adj_i(\hat{G}) \setminus \{j\}| \geq \ell$ are considered;}
		\UNTIL{for each $i,j$, $\adj_i(\hat{G}) \setminus \{j\} < \ell$.}
	\end{algorithmic}
\end{algorithm*}


 Similar to CMIT~\citep{Anand12}, Algorithm~\ref{alg:mtp2} starts with the fully connected graph $\hat{G}$ and sequentially removes edges based on conditional independence tests. The algorithm iterates with respect to a parameter $\ell$ that starts at $\ell = 0$. In each iteration, for all pairs of nodes $i, j$ such that the edge $(i,j) \in \hat{G}$ and node $i$ has at least $\ell$ neighbors (denoted by $\adj_i(\hat{G})$), the algorithm considers all combinations of subsets $S$ of $\adj_i(\hat{G})$ excluding $j$ that have size $\ell$ and all nodes $k \neq i,j $ that are not in $S$. For each combination of subset $S$ and node $k$, it calculates the empirical partial correlation coefficient $\hat{\rho}_{ij \mid S \cup \{ k \}}$. 
 Importantly, $\hat{\rho}_{ij \mid S \cup \{ k \}}$ 
 is calculated only on a \emph{subset} (which we refer to as a \emph{batch}) of size $M := N^\gamma$ that we draw randomly from the $N$ samples. 
 If any of these empirical partial correlation coefficients are negative, then edge $i - j$ is deleted from $\hat G$ (and no further tests are performed on $(i,j)$). Each iteration of the algorithm increases $\ell$ by~1 and the algorithm terminates when for all nodes $i,j$ such that $(i,j) \in \hat{G}$, the neighborhood of $i$ excluding $j$ has size strictly less than $\ell$. 

The basic intuition behind Algorithm~\ref{alg:mtp2} is that if there is an edge $i-j$ in $G$, then all partial correlations $\rho_{ij \mid S}$ are positive because of the basic properties of \M2. In the limit of large $N$, this implies that all $\hat{\rho}_{ij \mid S}$ are positive. On the other hand, when $i$ and $j$ are not connected in the true underlying graph, then there exists a list of conditioning sets $S_1, \cdots, S_K$ such that $\rho_{ij \mid S_k}=0$ for all $1\leq k\leq K$. When $K$ is large enough, then intuitively there should exist $1\leq k\leq K$ such that $\hat{\rho}_{ij \mid S_k} < 0$ with high probability. However, since for overlapping conditioning sets the empirical partial correlations are highly correlated, we use separate batches of data for their estimation. 
This leads to a procedure for learning the underlying Gaussian graphical model by deleting edges based on the signs of empirical partial correlation coefficients.

Having provided the high level intuition behind Algorithm~\ref{alg:mtp2}, we now prove its consistency under common assumptions on the underlying data generating process. Let $d$ denote the maximum degree of the true underlying graph $G$. For any positive semidefinite matrix $A$, let $\lambda_{\min}(A)$ and $\lambda_{\max}(A)$ denote the minimum and maximum eigenvalues of $A$ respectively. 
\begin{condition}\label{cd:eigen}
There exist positive constants $\sigma_{\min}$ and $\sigma_{\max}$ such that for any subset of nodes $S \subseteq [p]$ with $|S| \leq d + 4$, the true underlying covariance matrix satisfies
\begin{align*}
    \lambda_{\min}(\Sigma_S) \geq \sigma_{\min} \quad\textrm{and}\quad \lambda_{\max}(\Sigma_S) \leq \sigma_{\max}.
\end{align*}
\end{condition}

Note that since $\lambda_{\max}(\Sigma_S) \leq \tr(\Sigma_S)$ and $|S| \leq d + 4$, it is straightforward to show that a sufficient condition for $\lambda_{\max}(\Sigma_S)\leq\sigma_{\max}$ is that all diagonal entries of $\Sigma$ scale as a constant. This condition is also required by many existing methods including graphical lasso and CLIME; see Section~\ref{sec:pre}. 

Similarly, a sufficient condition for $\lambda_{\min}(\Sigma_S)\geq \sigma_{\min}$ is that all diagonal entries of $\Theta$ scale as a constant (see the Supplementary Material for a proof); this assumption is also required by CLIME. 

\begin{condition}\label{cd:signal}
There exists a positive constant $c_\rho$ such that for any two nodes $i, j \in [p]$, if $(i,j) \in G$, then $\rho_{i,j \mid [p] \setminus \{i,j\}} \geq c_\rho \sqrt{(\log p)/(N^{3 / 4})}$.
\end{condition}
Condition~\ref{cd:signal} is a standard condition for controlling the minimum edge weight in $G$ as required, for example, by graphical lasso. 
While the minimum threshold in our condition scales as $\sqrt{(\log p)/(N^{3 / 4})}$, graphical lasso only requires $\sqrt{(\log p)/N}$ (but instead requires a particular choice of tuning parameter and the incoherence condition). 

\begin{condition}\label{cd:size}
The size of $p$ satisfies that $p \geq N^{\frac{1}{8}} + d + 2$.
\end{condition}
Condition~\ref{cd:size} implies that the high-dimensional consistency guarantees of Algorithm~\ref{alg:mtp2} cannot be directly generalized to the low-dimensional setting where $p$ scales as a constant. 
We now provide the main result of our paper, namely consistency of Algorithm~\ref{alg:mtp2}. 

\begin{theorem} \label{thm:main}
Assume that the maximum neighbourhood size $d$ scales as a constant and let Conditions~\ref{cd:eigen}-\ref{cd:size} be satisfied with $c_\rho$ sufficiently large. Then for any $\gamma \in (\frac{3}{4}, 1)$, there exist positive constants $\tau$ and $C$ that depend on $(c_\rho, \sigma_{\max}, \sigma_{\min}, d, \gamma)$ such that with probability at least $1 - p^{-\tau} - p^2 e^{-C N^{\frac{1 - \gamma}{2} \wedge (4 \gamma - 3)}}$, the graph estimated by Algorithm~\ref{alg:mtp2} is the same as the underlying graph $G$.
\end{theorem}%


\begin{remark}\label{rmk:gamma}
The consistency guarantees of our algorithm hold for any $\gamma\in(\frac{3}{4},1)$. This means that our algorithm does not require tuning of the parameter $\gamma$ to consistently estimate the underlying graph $G$. 
Note that this is in contrast to other methods like graphical lasso or CLIME,
where the consistency guarantees require a specific choice of the tuning parameter in the algorithm, which is unknown a priori. This is advantageous, since our algorithm can consistently estimate the graph without running any computationally expensive tuning parameter selection approaches, such as stability selection~\citep{MB10}. By setting $\frac{1 - \gamma}{2} = (4 \gamma - 3)$, we obtain that the \emph{theoretically optimal} value is $\gamma=7/9$, as this leads to the best asymptotic rate. However, as seen in Section~\ref{sec:eval}, in practice different values of $\gamma$ can lead to different results. In particular, higher values of $\gamma$ empirically lead to removing less edges since the overlap between batches is higher and thus the empirical partial correlation coefficients are more correlated with each other.
\end{remark}

\begin{remark}
In applications where domain knowledge regarding the graph sparsity is available,  $\gamma$ can still be tuned to incorporate such knowledge to improve estimation accuracy. We see it as a benefit of our method that a tuning parameter can be used when one has access to domain knowledge, but doesn't have to be tuned in order to obtain consistent estimates, since it is provably consistent for all $\gamma\in(\frac{3}{4},1)$.
\end{remark}
%
%
%
\paragraph{Proof of Theorem~\ref{thm:main}:} In the following, we provide an overview of the proof of our main result. %
Theorems~\ref{thm:fn} and~\ref{thm:fp} show that at iteration $\ell = d + 1$, the graph $\hat{G}$ estimated by Algorithm~\ref{alg:mtp2} is exactly the same as the underlying graph $G$. The proof is then completed by showing that Algorithm~\ref{alg:mtp2} stops exactly at iteration $\ell = d + 1$. All proofs are provided in the Supplementary Material.

We start with Theorem~\ref{thm:fn}, which bounds the \emph{false negative rate} of Algorithm~\ref{alg:mtp2}, i.e.~showing that all edges $(i,j)$ in the true graph $G$ are retained. 

\begin{theorem}[False negative rate] \label{thm:fn}
Under Conditions~\ref{cd:eigen} and~\ref{cd:signal} and $c_\rho$ sufficiently large, there exists a positive constant $\tau$ that depends on $(c_\rho, \sigma_{\max}, \sigma_{\min}, d)$ such that with probability at least $1 - p^{-\tau}$, the graph $\hat G$ estimated by Algorithm~\ref{alg:mtp2} at iteration $\ell = d + 1$ contains all edges $(i,j) \in G$. 
\end{theorem}

The proof of Theorem~\ref{thm:fn} is based on concentration inequalities in estimating partial correlation coefficients. 
The high-level intuition behind the proof is that because the empirical partial correlation coefficients concentrate exponentially around the true partial correlation coefficients, then with high probability if an edge exists, no empirical partial correlation coefficient will be negative; as a consequence, Algorithm~\ref{alg:mtp2} will not eliminate the edge.

The following theorem bounds the \emph{false positive rate}; namely, it shows that with high probability Algorithm~\ref{alg:mtp2} will delete all edges $(i,j)$ that are not in the true graph $G$. 

\begin{theorem}[False positive rate]\label{thm:fp}
Under the same conditions as Theorem~\ref{thm:main}, there exists positive constants $C, \tau$ that depend on $(c_\rho, \sigma_{\max}, \sigma_{\min}, d, \gamma)$ such that with probability at least $1 - p^{-\tau} - p^2e^{-C\frac{1 - \gamma}{2} \wedge 4 \gamma - 3}$, the graph $\hat G$ estimated by Algorithm~\ref{alg:mtp2} at iteration $\ell = d + 1$ does not contain any edges $(i,j) \notin G$. 
\end{theorem}

The proof of Theorem~\ref{thm:fp} relies heavily on the following lemma that considers the orthant probability of partial correlation coefficients. Recall in Algorithm~\ref{alg:mtp2} that for a particular edge $i -j $ in the estimated graph $\hat G$ at a given iteration, we calculate a series of empirical partial correlation coefficients with different conditioning sets. The only way Algorithm~\ref{alg:mtp2} will not delete the edge is if all empirical partial correlation coefficients are $\geq 0$. Thus given 2 nodes $i, j$ for which $(i, j) \notin G$, we need to upper bound the orthant probability that all empirical partial correlation coefficients computed by Algorithm~\ref{alg:mtp2} are non-negative. As we will discuss next, the use of batches is critical for this result.

\begin{lemma}\label{lem:fn}
Consider a pair of nodes $(i,j)\notin G$. Assume that there exists $K := N^{\frac{1 - \gamma}{2}}$ sets of nodes $S_1, \cdots, S_K \subseteq [p] \setminus \{i,j\}$ with $|S_k| \leq d + 2$ that satisfy $\rho_{ij \mid S_k} = 0$. Then there exists positive constants $C$ and $N_0$ that depends on $(\sigma_{\max}, \sigma_{\min},d)$ such that
\begin{align}\label{eq:exp}
\Pr(\hat{\rho}_{ij \mid S_k} > 0\quad  \forall k \in [K]) \leq \exp(- C N^{\frac{1 - \gamma}{2} \wedge 4 \gamma - 3}).
\end{align}
\end{lemma}

\begin{figure*}[!t]
	\centering
	\subcaptionbox{Random graphs}{\includegraphics[width=0.31\textwidth]{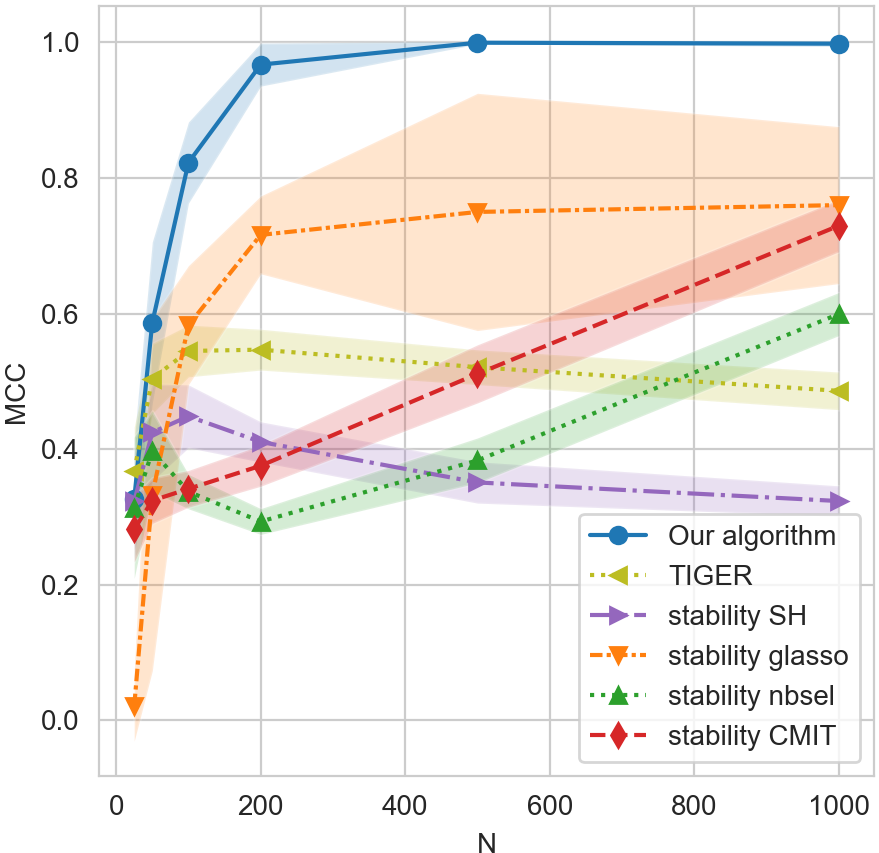}}%
	\hfill 
	\subcaptionbox{Chain graphs}{\includegraphics[width=0.31\textwidth]{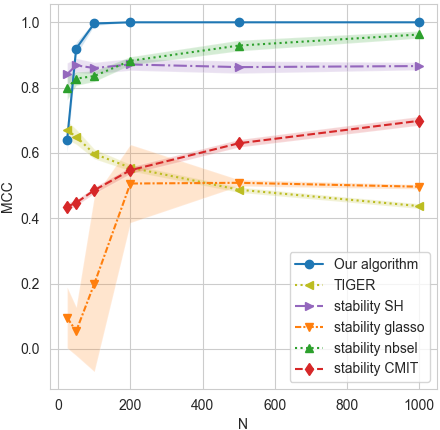}}%
	\hfill
	\subcaptionbox{Grid graphs}{\includegraphics[width=0.31\textwidth]{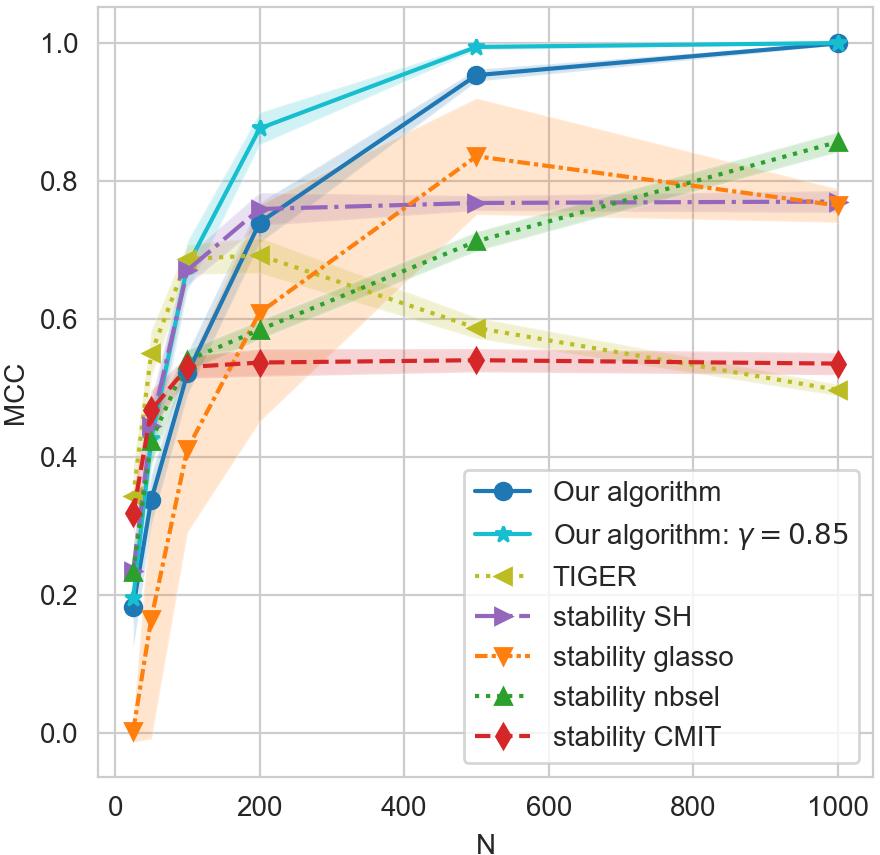}}%
	\caption{Comparison of different algorithms evaluated on MCC across (a) random, (b) chain, (c) grid graphs with $p=100$ and $N \in \{ 25, 50, 100, 200, 500, 1000\}$. For each graph and choice of $p$ and $N$, results are shown as an average across $20$ trials. The shaded areas correspond to $\pm 1$ standard deviation of MCC over $20$ trials.}
	\label{fig:mcc}
	\vspace{-0.2cm}
\end{figure*}

To provide intuition for the proof of Lemma~\ref{lem:fn}, consider a scenario where the batch size $M$ is chosen small enough such that the batches used to estimate the different $\hat{\rho}_{ij \mid S_k}$'s have no overlap. Since in this case all $\hat{\rho}_{ij \mid S_k}$'s are independent, the bound in Lemma~\ref{lem:fn} can easily be proven, namely: for some positive constant $\delta < 1$, it holds that
\begin{align*}
    \Pr(\hat{\rho}_{ij \mid S_k} & > 0  \quad \forall k \in [K])  = \prod_{k=1}^K \Pr(\hat{\rho}_{ij \mid S_k} > 0) \\
    & \leq \delta^K = \exp\big(-\log (1 / \delta) \cdot N^{\frac{1 - \gamma}{2}}\big).
\end{align*}
However, for small batch size $M$ the empirical partial correlation coefficients $\hat{\rho}_{ij \mid S}$ don't concentrate around $\rho_{ij \mid S}$, which may result in false negatives. In the proof of Lemma~\ref{lem:fn} we show that choosing a batch size of $M = N^\gamma$ guarantees the required concentration result as well as a sufficiently weak dependence among the empirical partial correlation coefficients $\hat{\rho}_{ij \mid S_k}$'s to obtain the exponential upper bound in~\eqref{eq:exp} as in the independent case.
Lemma~\ref{lem:fn} implies Theorem~\ref{thm:fp} by taking uniform control over all edges $(i,j) \not\in G$. 
%
Finally, to complete the proof of Theorem~\ref{thm:main}, it remains to show that Algorithm~\ref{alg:mtp2} terminates at  iteration $\ell = d + 1$. 

\begin{proof}[Proof of Theorem~\ref{thm:main}]
It follows from Theorem~\ref{thm:fn} and Theorem~\ref{thm:fp} that with probability at least $1 - p^{-\tau} - p^2 e^{-CN^{\frac{1 - \gamma}{2} \wedge 4 \gamma - 3}}$, the graph estimated by Algorithm~\ref{alg:mtp2} at iteration $\ell = d + 1$ is exactly the same as $G$. Since the maximum degree of $G$ is at most $d$, it matches the stopping criterion of Algorithm~\ref{alg:mtp2}. As a consequence, Algorithm~\ref{alg:mtp2} terminates at iteration $\ell = d + 1$. 
\end{proof}





\section{Empirical Evaluation}
\label{sec:eval}

In the following, we evaluate the performance of our algorithm for structure recovery in \M2 Gaussian graphical models in the high-dimensional, sparse regime. We first compare the performance of Algorithm~\ref{alg:mtp2} to various other methods on synthetically generated datasets and then present an application to graphical model estimation on financial data. The code to reproduce our experimental results is available at \url{https://github.com/puma314/MTP2-no-tuning-parameter}.

\subsection{Synthetic Data}
\label{sec_synthetic}

\begin{figure*}[!t]
	\centering
	\subcaptionbox{ROC curve
	}{\includegraphics[width=0.3\textwidth]{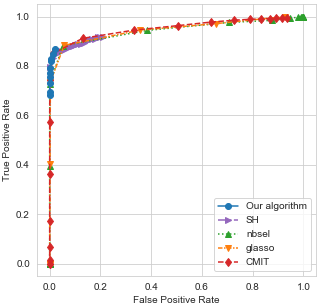}}
	\hfill
	\subcaptionbox{MCC}
	{\includegraphics[width=0.3\textwidth]{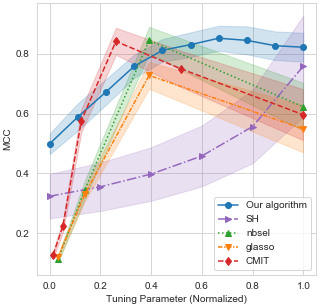}}%
	\hfill
	\subcaptionbox{True positive rate}
	{\includegraphics[width=0.3\textwidth]{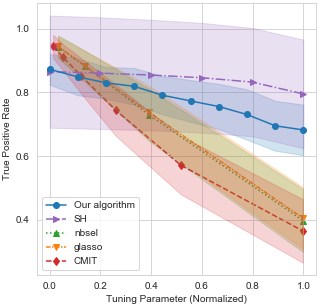}}%
	\caption{(a) ROC curves, (b) MCC, and (c) true positive rate versus normalized tuning parameter for random graphs with $p=100$ and $N=500$ across $30$ trials. The shaded regions correspond to $\pm 1$ standard deviation of MCC (TPR resp.) across $30$ trials.}
	\label{fig:roc}
	\vspace{-0.2cm}
\end{figure*}

Given a precision matrix $\Theta \in \R^{p \times p}$, we generate $N$ i.i.d. samples $x^{(1)}, \ldots, x^{(N)} \sim \mathcal{N}(0, \Theta^{-1})$. We let $\hat{\Sigma} = \frac{1}{N} \sum_{i=1}^{N} (x^{(i)})(x^{(i)})^T$ denote the sample covariance matrix. To analyze the performance of our algorithm in various scenarios, we vary $N$ for $p=100$. In addition, we consider three different sparsity patterns in the underlying precision matrix $\Theta$ that are similarly considered by~\citet{SH15}, namely: 

\emph{Grid: } Let $B$ be the adjacency matrix of a 2d-grid of size $\sqrt{p}$. Let $\delta := 1.05 \cdot \lambda_1(B)$, $\tilde{\Theta} := \delta I - B$ and $\Theta = D \tilde{\Theta} D$, where $D$ is a diagonal matrix such that $\Sigma = \Theta^{-1}$ has unit diagonal entries.

\emph{Random: } Same as for \emph{grid} above, but with $B$ replaced with a symmetric matrix having $0$ diagonal and one percent non-zero off diagonal entries uniform on $[0,1]$ chosen uniformly at random.

\emph{Chain: } We let $\Sigma^* := (\sigma^*_{jk}) = (0.9^{|j-k|}), j,k = 1, \ldots, p$. Then we take $\Omega := (\Sigma^*)^{-1}$.

Our primary interest in comparing different algorithms is their performance at recovering the underlying graph structure associated with $\Theta$. 
Similarly as in~\citep{SH15}, in Figure~\ref{fig:mcc} we evaluate their performance using Matthew's correlation coefficient (MCC):
\begin{align*}
& \text{MCC} =\\
& \frac{\TP \cdot \TN - \FP \cdot \FN}{ \left( (\TP + \FP)(\TP + \FN)(\TN + \FP)(\TN + \FN)\right)^{1/2}},
\end{align*}
where $\TP$, $\TN$, $\FP$ and $\FP$ denote the number of true positives, true negatives, false positives and false negatives respectively. Intuitively, MCC measures the correlation between the presence of edges in the true and estimated graphs. Thus, a higher MCC score means less number of false positives \emph{and} false negatives. Since MCC combines true positive rates (TPR) and false positive rates (FPR), we think it is a compelling metric. MCC has also been used in similar work~\citep{SH15}. In the appendix, we also provide evaluation results based on TPR and FPR.

\emph{Choice of Parameters:} We fix $p=100$ and vary $N = 25, 50, 100, 200, 500, 1000$ to analyze how the ratio $p/N$ affects performance for the various algorithms. For each setup and value of $N$, we do $20$ trials of each algorithm and report the average of the MCCs across the trials.


\emph{Methods Compared:} We benchmark our algorithm against a variety of state-of-the-art methods for structure learning in Gaussian graphical models (see Section~\ref{sec:pre}) for a range of tuning parameters: 
\begin{itemize}
    \item {SH:} Slawski and Hein~\citep{SH15} considered the same problem as in this paper. For comparison to their algorithm we use the same range of tuning parameters as considered by them, namely $q \in \{ 0.7, 0.75, 0.8, 0.85, 0.9, 0.95, 0.99\}$. 
    \item \emph{glasso:} For graphical lasso~\citep{FHT08} we vary the sparsity parameter around the the theoretically motivated tuning parameter of $\sqrt{\log(p)/n}$, namely $\lambda \in \{ 0.055,  0.16,  0.45,  1.26,  3.55, 10 \}$.
    \item \emph{nbsel:} For neighborhood selection~\citep{MB06} we use the same $\lambda$ values as for \emph{glasso}.
    \item \emph{TIGER:} For TIGER~\citep{LWang17}, we use the theoretically optimal value $\lambda := \pi \sqrt{\frac{\log (p)}{n}}$.
    \item \emph{CMIT: } This algorithm~\citep{Anand12} has two tuning parameters. Since the run-time is $p^{\eta+2}$ in the maximal size of the conditioning set $\eta$, we set $\eta = 1$ for computational reasons. For $\lambda$, we use the the same values as for \emph{glasso}.
    \item \emph{Our algorithm:} We use the asymptotically optimal choice of $\gamma = 7/9$ (see Remark \ref{rmk:gamma}) and also compare to $\gamma = 0.85$, which falls in the allowable range $(0.75,1)$.
\end{itemize}

For the comparison based on MCC in Figure~\ref{fig:mcc}, we use \emph{stability selection} \citep{MB10}, where an algorithm is run multiple times with different subsamples of the data for each tuning parameter and an edge is included in the estimated graph if it is selected often enough (we used $80\%$).

\emph{Discussion: } Figure~\ref{fig:mcc} compares the performance of the various methods based on MCC for random graphs, chain graphs and grid graphs. Compared with the algorithm that has similar theoretical properties as ours, namely TIGER, our algorithm has better overall performance across all simulation set-ups. For the other state-of-the-art methods, Figure~\ref{fig:mcc}(a) shows that our algorithm is able to offer a significant improvement for random graphs over competing methods. Also on chain graphs (Figure~\ref{fig:mcc}(b)) our algorithm is competitive with the other algorithms, with \emph{SH} and \emph{nbsel} performing comparably. For the grid graph (Figure~\ref{fig:mcc}(c)), for $N \leq 200$ \emph{SH} with stability selection outperforms our algorithm with $\gamma = 7/9$. However, it is important to note that stability selection is a major advantage for the compared algorithms and comes at a significant computational cost.  Moreover, by varying $\gamma$ in our algorithm its performance can be increased and becomes competitive to \emph{SH} with stability selection. Both points are discussed in more detail in the Supplementary Material. Another interesting phenomenon is that in Figure~\ref{fig:mcc}(c), our algorithm with $\gamma = 0.85$ performs better than the ``theoretically optimal'' $\gamma = 7 / 9$, which may seem to contradict our theoretical results. Notice, however, that ``theoretical optimality'' holds for $N\to\infty$. In the finite sample regime considered here factors such as $\sigma_{\min}$, $\sigma_{\max}$ and $d$ can influence the optimal choice.

To evaluate the sensitivity of the various algorithms to their respective tuning parameters, 
 we generate an ROC curve for each algorithm on random graphs with $p=100$ and $N \in \{25, 50, 100, 200, 500, 1000\}$, of which $N=500$ is shown in Figure~\ref{fig:roc}(a); see the Supplementary Material for more details and plots. All algorithms  perform similarly in terms of their ROC curves. Note that since our algorithm can only choose $\gamma$ from the range $(0.75, 1)$, its false positive rate is upper bounded and thus it is impossible to get a full ``ROC'' curve. 
 Figure~\ref{fig:roc}(b) and (c) show the MCC and true positive rate (TPR) for each algorithm as a function of the tuning parameter normalized to vary between $[0,1]$. 
 Our algorithm is the least sensitive to variations in the tuning parameter, as it has one of the smallest ranges in both MCC and TPR (the $y$-axes) as compared to the other algorithms. Our algorithm also shows the smallest standard deviations in MCC and in TPR, showing its consistency across trials (especially compared to \emph{SH}). We here concentrate on TPR since the variation in FPR between  all algorithms is small across trials. 
 Taken together, it is quite striking that our algorithm with fixed $\gamma$ generally outperforms methods with stability selection.

\subsection{Application to Financial Data}

We now examine an application of our algorithm to financial data. The \M2 constraint is relevant for such data, since the presence of a latent global market variable leads to positive dependence among  stocks~\citep{HL02,MS05}. We consider the daily closing prices for $p=452$ stocks that were consistently in the S\& P 500 index from January 1, 2003 to January 1, 2018, which results in a sample size of $N=1257$. Due to computational limitations of stability selection primarily with \emph{CMIT}, we performed the analysis on the first $p = 100$ of the $452$ stocks. 
The $100$ stocks are categorized into 10 sectors, known as the Global Industry Classification Standard (GICS) sectors. This dataset is gathered from Yahoo Finance and has also been analyzed in~\citep{LHZ12}.

A common task in finance is to estimate the covariance structure between the log returns of stocks. Let $S_j^{(t)}$ denote the closing price of stock $j$ on day $t$ and let $X_j^{(t)} := \log ( S_j^{(t)} / S_j^{(t-1)} )$ denote the log return of stock $j$ from day $t-1$ to $t$. Denoting by $X := (X_1, \ldots, X_{100})^T$ the random vector of daily log returns of the $100$ stocks in the data set, then our goal is to estimate the undirected graphical model of $X$. We do this by treating the $1257$ data points $X^{(t)} := (X_1^{(t)}, \ldots, X_{100}^{(t)})$ corresponding to the days $t=1, \ldots, 1257$ as i.i.d.~realizations of the random vector $X$. 

As in Section~\ref{sec_synthetic}, we compare our method to \emph{SH}, \emph{glasso} (using both stability selection and cross-validation), \emph{nbsel}, \emph{CMIT} (using both stability selection and the hyperparameter with the best performance) and TIGER. 
Note that here we cannot assess the performance of the various methods using MCC since the graph structure of the true underlying graphical model is unknown. Instead, we assess each estimated graph based on its \emph{modularity coefficient}, namely the performance at grouping stocks from the same sector together. Table~\ref{tab:modularity} shows that our method using fixed $\gamma = 7 / 9$ outperforms all other methods in grouping the stocks. For further details on the analysis see the Supplementary Material.

\begin{table}[t!]
	\centering
	\begin{tabular}{ c c }
		Method & Modularity \\ & Coefficient \\ \hline \hline
		Our Algorithm $(\gamma = 7./9.$) & 0.482 \\ \hline
		Slawski-Hein with st.~sel. & 0.418 \\ \hline
		Neighborhood selection with st.~sel. & 0.350\\ \hline
		Graphical Lasso with st.~sel. & 0. \\ \hline
		Cross-validated graphical lasso & 0.253\\ \hline
		\emph{CMIT} with st.~sel. & -0.0088  \\ \hline
		\emph{CMIT} with best hyperparameter & -0.0085 \\ \hline
		TIGER & -0.5 \\ 
	\end{tabular}
	\vspace{-0.1cm}
	\caption{Modularity scores of the estimated graphs; higher score indicates better clustering performance; ``st.~sel'' stands for ``stability selection''. For our algorithm we used the theoretically optimal value of $\gamma=7/9$.}\label{tab:modularity}
	\vspace{-0.3cm}
\end{table}

\section{Discussion}
\label{sec:discuss}

In this paper, we proposed a tuning-parameter free, constraint-based estimator for learning the structure of the underlying Gaussian graphical model under the constraint of \M2. We proved consistency of our algorithm in the high-dimensional setting without relying on an unknown tuning parameter. We further benchmarked our algorithm against existing algorithms in the literature with both simulated and real financial data, thereby showing that it outperforms existing algorithms in both settings. A limitation of our algorithm is that its time complexity scales as $O(p^d)$; it would be interesting in future work to develop a more computationally efficient algorithm for graphical model estimation under \M2. Another limitation is that our algorithm is only provably consistent in the high-dimensional setting. However, the strong empirical performance of our algorithm as compared to existing algorithms is quite striking, given in particular these results are from fixed $\gamma$. To our knowledge, this is the first tuning-parameter free algorithm for structure recovery in Gaussian graphical models with consistency guarantees.


\subsubsection*{Acknowledgements}

We thank Dheeraj Nagaraj, Cheng Mao and Philippe Rigollet and the anonymous reviewers for helpful discussions. The authors acknowledge support by NSF (DMS-1651995), ONR (N00014-17-1-2147 and N00014-18-1-2765), IBM, a Sloan Fellowship and a Simons Investigator Award. At the time this research was completed, Yuhao Wang and Uma Roy was at the Massachusetts Institute of Technology.

\bibliography{references}

\newpage
\onecolumn

\appendix

\counterwithin{figure}{section}
\counterwithin{table}{section}

\section{Additional discussion of Condition~\ref{cd:eigen}}

In this section, we explain why a sufficient condition for ``$\lambda_{\min} (\Sigma_S) \geq \sigma_{\min}$'' is that all diagonal entries of $\Theta$ scale as constants:

When all diagonal entries of $\Theta$ scale as constants, standard results on the Schur complement yield that all diagonal entries in $(\Sigma_S)^{-1}$ also scale as constants. Hence, $\lambda_{\max} ((\Sigma_S)^{-1}) \leq \textrm{trace}((\Sigma_S)^{-1}) = \sum_{i \in S}[(\Sigma_S)^{-1}]_{ii}$ is also upper bounded by a constant (since $|S| \leq d + 4$). By combining this with the fact that $\lambda_{\min}(\Sigma_S) = \lambda_{\max}^{-1}(\Sigma_S^{-1})$, we can conclude that $\lambda_{\min}(\Sigma_S) \geq \sigma_{\min}$ for a positive constant $\sigma_{\min}$.

\section{Proof of Lemma~\ref{lem:fn}}

\subsection{Characterization of maximal overlaps}

Our proof of Lemma~\ref{lem:fn} relies on the following lemma that characterizes the size of maximal overlaps between any two batches.

\begin{lemma}[Tail-bounds on maximum overlap of subsets] \label{lem:max-overlap}
Consider a set of data $B := \{x^{(i)}\}_{i=1}^N$ with size $N$. Let $B_1, \cdots, B_K \subseteq B$ denote $K$ subsets where each $B_k$ is created by uniformly drawing $M$ samples from the set $B$, then $\forall \epsilon > 0$,
\begin{align*}
\Pr \left( \max_{i, j} |B_i \cup B_j| < \frac{M^2}{N} + \epsilon N \right) \geq 1 - \exp(-2 \epsilon^2 N + 2 \log K).
\end{align*}

\end{lemma}

\begin{proof}
By union bound, we have for any $T > 0$,

\begin{equation} \label{eq:union-bound}
\Pr(\max_{i, j} |B_i \cap B_j| > T) \leq \binom{K}{2} \Pr(|B_i \cap B_j| > T).
\end{equation}

For any $i \neq j$, let the random variable $y_\ell := \mathbf{1} \{ x^{(\ell)} \in B_i \} \cdot \mathbf{1} \{ x^{(\ell)} \in B_j \}$, it follows that $|B_i \cap B_j| = \sum_{l = 1}^{N} y_\ell$ and thus 
\begin{equation} \label{eq:set-trans}
\Pr \left( |B_i \cap B_j| > T \right) = \Pr \left( \sum_{\ell = 1}^{N} y_\ell > T \right).
\end{equation}
In addition, $y_\ell$ is a binary variable satisfying $\Pr(y_{\ell} = 1) = \left( \frac{M}{N} \right)^2$. 

In this case, it suffices to provide an upper bound on the probability $\Pr \left( \sum_{\ell = 1}^{N} y_\ell > T \right)$. Using basic results in combinatorics, one can rewrite the conditional probability $\Pr(y_\ell = 1 | y_{\ell'} = 1)$ as follows:
\begin{equation*}
\Pr(y_\ell = 1 | y_{\ell'} = 1) = \frac{|\{B_i : x^{(\ell')}, x^{(\ell)} \in B_i \}| \cdot |\{B_j : x^{(\ell')}, x^{(\ell)} \in B_j \}|}{|\{B_i : x^{(\ell')} \in B_i \}| \cdot |\{B_j : x^{(\ell')} \in B_j \}|}= \frac{\binom{N-2}{M-2}^2}{\binom{N-1}{M-1}^2}.
\end{equation*}
It follows that
\begin{equation*}
\Pr(y_\ell = 1 | y_{\ell'} = 1) = \frac{\binom{N-2}{M-2}^2}{\binom{N-1}{M-1}^2} = \left(\frac{M-1}{N-1}\right)^2 \leq \left(\frac{M}{N}\right)^2 = \Pr(y_\ell = 1),
\end{equation*}
which means for any $\ell \neq \ell'$, the random variables $y_\ell$ and $y_{\ell'}$ are negatively correlated. By applying Chernoff-Hoeffding bounds on sum of negatively associated random variables (see e.g.~\cite[Theorem~14]{DPR96}), we obtain
\begin{equation}\label{eq:hoeffding}
\Pr \left(\sum_{\ell = 1}^{N} (y_\ell - \E(y_\ell)) > \epsilon N \right) \leq \exp(-2 \epsilon^2 N).
\end{equation}
Combining \eqref{eq:union-bound}, \eqref{eq:set-trans} and \eqref{eq:hoeffding} and that $\E \left(y_\ell \right) = \frac{M^2}{N^2}$, we obtain the statement in the lemma.
\end{proof}

\subsection{Proof of Lemma~\ref{lem:fn}}

\noindent {\bf Notations and proof ideas for Lemma~\ref{lem:fn}.} To simplify notation, we denote each $\hat{\rho}_{ij \mid S_k}$ as $\hat{\rho}_k$ and denote the subset of data points used to estimate $\hat{\rho}_k$ as $B_k$. Let $\hat{\Sigma}_k \in \R^{|S_k| + 2 \times |S_k| + 2}$ denote the sample covariance matrix of the nodes $S_k \cup \{i,j\}$. Note that here $\hat{\Sigma}_k$ is estimated from the data in $B_k$. Let $\hat{\sigma}_k$ denote the vectorized form of $\hat{\Sigma}_k$ and let $\sigma_k$ denote the expectation of $\hat{\sigma}_k$. Standard results in calculating partial correlation coefficients show that $\hat{\rho}_k$ can be taken as a function of $\hat{\sigma}_k$, which we denote as
\begin{align*}
\hat{\rho}_k = g_k(\hat{\sigma}_k).
\end{align*}
Moreover, since the derivatives of all orders of $g_k(\cdot)$ at the point $\hat{\sigma}_k$ can be expressed as polynomials of $\hat{\sigma}_k$ and its inverse (see e.g. Eq.~36 in \citet{WKR14} and the two equations after that), $g_k(\cdot)$ is infinitely differentiable whenever the inputs are non-singular matrices. Let $\ell_k$ denote the first order derivative of $g_k$ at the point $\sigma_k$. It follows that $\ell_k (\hat{\sigma}_k - \sigma_k)$ is the first order approximation of $g_k(\hat{\sigma}_k)$. Let the residual 
\begin{align}\label{eq:fn2}
r_k := g_k(\hat{\sigma}_k) - \ell_k(\hat{\sigma}_k - \sigma_k).
\end{align}
Let $\|\hat{\sigma}_k - \sigma_k\|_\infty$ denote the $\ell_\infty$ norm of the vector $\hat{\sigma}_k - \sigma_k$. Standard results in Taylor expansion show that when $\|\hat{\sigma}_k - \sigma_k\|_\infty$ is negligible, one can rewrite the residual as
\begin{align*}
r_k = \frac{1}{2}(\hat{\sigma}_k - \sigma_k)^T H_k(\tilde{\sigma}_k) (\hat{\sigma}_k - \sigma_k),
\end{align*}
where $H_k(\cdot)$ is the Hessian matrix of $g_k$ and $\tilde{\sigma}_k$ is some point in the middle between $\hat{\sigma}_k$ and $\sigma_k$. Let $\brh := (\hat{\rho}_1, \cdots, \hat{\rho}_K)^T$, $\bL := (\ell_1(\hat{\sigma}_1 - \sigma_1), \cdots,\ell_K(\hat{\sigma}_K - \sigma_K))^T$ and $\bR := (r_1, \cdots, r_K)^T$. Since each $\hat{\sigma}_k$ is estimated using a subset of data with batch size $M$, there may be overlaps between the set of data used to calculate different $\hat{\sigma}_k$'s. Let $\hat{\sigma}_k^{(1)}$ denote the sample covariance matrix estimated from the data in $B_k \setminus \big(\underset{k' \neq k}{\cup} B_{k'}\big)$ and let $\hat{\sigma}_k^{(2)}$ denote the sample covariance matrix estimated from the data in the overlaps, i.e., the data in $B_k \cap \big(\underset{k' \neq k}{\cup} B_{k'}\big)$. Then one can decompose $\hat{\sigma}_k$ as $\hat{\sigma}_k = \frac{M-T_k}{M} \hat{\sigma}_k^{(1)} + \frac{T_k}{M} \hat{\sigma}_k^{(2)}$, where $T_k$ is the size of data in the overlaps. It is obvious that the $\hat{\sigma}_k^{(1)}$'s are independent from each other. Based on the above decomposition, we denote $\bL = \bL^{(1)} + \bL^{(2)}$, where 
\begin{align*}
\bL^{(1)} := \Big(\frac{M-T_k}{M} \ell_1(\hat{\sigma}_1^{(1)} - \sigma_1), \cdots, \frac{M-T_k}{M} \ell_K(\hat{\sigma}_K^{(1)} - \sigma_K)\Big)^T
\end{align*}
and
\begin{align*}
\bL^{(2)} := \Big(\frac{T_k}{M} \ell_1(\hat{\sigma}_1^{(2)} - \sigma_1), \cdots, \frac{T_k}{M} \ell_K(\hat{\sigma}_K^{(2)} - \sigma_K)\Big)^T.
\end{align*}
In addition, for any vector $\mathbf{a}$, we write $\mathbf{a} \geq 0$ whenever all elements of the vector $\mathbf{a}$ are greater than or equal to zero.

Let the random event 
\begin{align*}
\mathcal{B} := \Big\{(B_1, \cdots, B_K) : \max_{k, k' \in [K]} |B_k \cap B_{k'}| \leq 2 \frac{M^2}{N}\Big\}.
\end{align*}
By applying Lemma~\ref{lem:max-overlap}, it follows that there exists some positive constant $C$ that depends on $\gamma$ such that $\Pr(\mathcal{B}) \geq 1 - \exp(-C N^{4 \gamma - 3})$. By combining this with the decomposition, we have
\begin{align*}
\Pr(\brh \geq 0) & = \Pr(\brh \geq 0, \mathcal{B}) + \Pr(\brh \geq 0, \neg\mathcal{B}) \leq \Pr(\brh \geq 0 \mid \mathcal{B})\Pr(\mathcal{B}) + \Pr(\neg\mathcal{B}) \\
& \leq \Pr(\brh \geq 0 \mid \mathcal{B}) + \Pr(\neg\mathcal{B}),
\end{align*}
where $\neg\mathcal{B}$ denotes the complement of the random event $\mathcal{B}$. It is sufficient to prove Lemma~\ref{lem:fn} by proving
\begin{align}\label{eq:fn1}
\Pr(\brh \geq 0 \mid \mathcal{B}) \leq \exp(-C N^{\frac{1 - \gamma}{2}})
\end{align}
for some positive constant $C$ that depends on $\sigma_{\max}$, $\sigma_{\min}$ and $d$. In other words, it remains to prove that $\Pr(\brh \geq 0) \leq \exp(-C N^{\frac{1 - \gamma}{2}})$ when we are under a \emph{particular} subsampling assignment $(B_1, \cdots, B_K)$ that is in the random event $\mathcal{B}$.

\noindent {\bf Preliminary lemmas for Lemma~\ref{lem:fn}.}

Since the only remaining task is to deal with Eq.~\eqref{eq:fn1}, for the remainder of the proof of Lemma~\ref{lem:fn} we can assume that we are under a \emph{particular} subsampling assignment $(B_1, \cdots, B_K)$ in $\mathcal{B}$. To simplify notation we omit ``$\mid \mathcal{B}$'' in the remainder of the proof.

\begin{lemma}\label{lem:concentration}
For all $\epsilon > 0$, there exists some positive constant $C$ that depends on $d$, $\sigma_{\max}$ and $\sigma_{\min}$ such that the following inequality holds:
\begin{align*}
\Pr(\|\hat{\sigma}_k - \sigma_k\|_\infty > \epsilon) \leq 2 (d + 2)^2 e^{-CM\epsilon^2}.
\end{align*}
\end{lemma}
\begin{proof}
This is a direct consequence of~\cite[Lemma~7]{WKR14} and the Gaussianity of the underlying distribution.
\end{proof}
\begin{lemma}\label{lem:residual}
For all $\epsilon > 0$, there exist positive constants $C_1$ and $C_2$ that depend on $\sigma_{\max}$, $\sigma_{\min}$ and $d$ such that
\begin{align*}
\Pr(\|\bR\|_\infty \leq \epsilon) \geq 1 - 2 (d+2)^2 e^{\frac{1 - \gamma}{2} \log N - C_1 M \epsilon} - 2 (d+2)^2 e^{\frac{1 - \gamma}{2} \log N - C_2 \sqrt{M}}.
\end{align*}
\end{lemma}

\begin{proof}
For each $r_k$, let $C_1 - C_3$ denote a positive constant that depends on $\sigma_{\min}$, $\sigma_{\max}$ and $d$ and may vary from line to line. We have that
\begin{align}\label{eq:residual1}
\Pr(|r_k| > \epsilon) &= \Pr(|r_k| > \epsilon, \|\hat{\sigma}_k - \sigma_k\|_\infty \leq M^{-1/4}) + \Pr(|r_k| > \epsilon, \|\hat{\sigma}_k - \sigma_k\|_\infty \geq M^{-1/4}) \nonumber\\
&\leq \Pr((|r_k| > \epsilon, \|\hat{\sigma}_k - \sigma_k\|_\infty \leq M^{-1/4}) + \Pr(\|\hat{\sigma}_k - \sigma_k\|_\infty \geq M^{-1/4}).
\end{align}
Under the random event where $\|\hat{\sigma}_k - \sigma_k\|_\infty \leq M^{-1/4}$, standard results in Taylor expansion show that $r_k$ can be expressed in the form $r_k = (\hat{\sigma}_k - \sigma_k)^T H_k(\tilde{\sigma}_k) (\hat{\sigma}_k - \sigma_k)$. Thus one can rewrite~\eqref{eq:residual1} as
\begin{align*}
\Pr(|r_k| > \epsilon) & \leq \Pr\Big(\Big|\frac{1}{2}(\hat{\sigma}_k - \sigma_k)^T H_k(\tilde{\sigma}_k) (\hat{\sigma}_k - \sigma_k)\Big| > \epsilon, \|\hat{\sigma}_k - \sigma_k\|_\infty \leq M^{-1/4}\Big) \\
&\quad\quad + \Pr(\|\hat{\sigma}_k - \sigma_k\|_\infty \geq M^{-1/4}).
\end{align*}
Under the random event $\|\hat{\sigma}_k - \sigma_k\|_\infty \leq M^{-1/4}$, $\tilde{\sigma}_k$ is in the middle of $\hat{\sigma}_k$ and $\sigma_k$. It follows that $\|\tilde{\sigma}_k - \sigma_k\|_\infty \leq M^{-1/4}$. By combining this with the fact that the Hessian function $H_k(\cdot)$ is infinitely differentiable at the point $\sigma_k$, there exists some positive constant $C_1$ such that $\|H_k(\tilde{\sigma}_k) - H_k(\sigma_k)\|_\infty \leq C_1$. Using that $\|H_k(\sigma_k)\|_\infty$ is also bounded by a positive constant (since it is a function of $\sigma_k$, see e.g. \cite[Section~6.5]{WKR14} and \cite[Page 185]{MN88} for the explicit form), we further obtain that $\|H_k(\tilde{\sigma}_k)\|_\infty \leq C_1$. As a consequence, one can further rewrite~\eqref{eq:residual1} as
\begin{align*}
\Pr(|r_k| > \epsilon) &\leq \Pr(d \sqrt{C_1} \|\hat{\sigma}_k - \sigma_k\|_\infty > \sqrt{\epsilon}, \|\hat{\sigma}_k - \sigma_k\|_\infty \leq M^{-1/4}) \\
& \quad\quad + \Pr(\|\hat{\sigma}_k - \sigma_k\|_\infty \geq M^{-1/4})\\
& \leq \Pr(d\sqrt{C_1} \|\hat{\sigma}_k - \sigma_k\|_\infty > \sqrt{\epsilon}) + \Pr(\|\hat{\sigma}_k - \sigma_k\|_\infty \geq M^{-1/4}).
\end{align*}
By applying Lemma~\ref{lem:concentration}, we conclude that $\Pr(|r_k| > \epsilon) \leq 2 (d+2)^2 e^{-C_2 M \epsilon} + 2 (d+2)^2 e^{-C_3 \sqrt{M}}$. By taking the union bound over all $k \in [K]$, we obtain the desired statement in the lemma.
\end{proof}
\begin{lemma}\label{lem:overlap-cov}
Let $T := \max_k T_k$. For all $\epsilon > 0$, there exists some positive constant $C$ that depends on $\sigma_{\max}$, $\sigma_{\min}$ and $d$ such that
\begin{align*}
\Pr(\|\bL^{(2)}\|_\infty \leq \epsilon) \geq 1 - 2(d + 2)^2 e^{\frac{1 - \gamma}{2} \log N - C \frac{M^2}{T} \epsilon^2}.
\end{align*}
\end{lemma}
\begin{proof}
For each $\hat{\sigma}_k^{(2)}$, it follows from Lemma~\ref{lem:concentration} that for all $\epsilon > 0$, there exists some positive constant $C$ that depends on $\sigma_{\max}, \sigma_{\min}$ as well as $d$ such that
\begin{align*}
\Pr(|\ell_k (\hat{\sigma}_k^{(2)} - \sigma_k)| > \epsilon) \leq \Pr(\|\ell_k\|_1 \|\hat{\sigma}_k^{(2)} - \sigma_k\|_\infty > \epsilon) \leq 2(d + 2)^2 e^{-C T_k \epsilon^2},
\end{align*}
where the term $\|\ell_k\|_1$ is absorbed into the positive constant $C$ since $\|\ell_k\|_1$ is a constant that depends on $\sigma_{\max}, \sigma_{\min}$ and $d$. By taking the union bound and using that $T_k \leq T$, we obtain
\begin{align*}
\Pr(\|\bL^{(2)}\|_\infty > \epsilon) & \leq \sum_{k=1}^K \Pr\big( \frac{T_k}{M} |\ell_k (\hat{\sigma}_k^{(2)} - \sigma_k)| > \epsilon\big) \leq 2(d + 2)^2 N^{\frac{1 - \gamma}{2}} e^{-C \frac{M^2}{T_k} \epsilon^2} \\
& \leq 2(d + 2)^2 N^{\frac{1 - \gamma}{2}} e^{-C \frac{M^2}{T} \epsilon^2},
\end{align*}
which completes the proof.
\end{proof}

With these preparations we can now prove Lemma~\ref{lem:fn}.
\begin{proof}[Proof of Lemma~\ref{lem:fn}]
Let $C_1-C_6$ denote positive constants that depend on $\sigma_{\min}$, $\sigma_{\max}$ and $d$ and may vary from line to line. For any $\epsilon > 0$, standard results in probability yield that
\begin{align*}
\Pr(\brh \geq 0) & = \Pr(\brh \geq 0, \|\bR\|_\infty \leq \epsilon) + \Pr(\brh \geq 0, \|\bR\|_\infty \geq \epsilon) \nonumber \\
& \leq \Pr(\bL + \bR \geq 0, \|\bR\|_\infty \leq \epsilon) + \Pr(\|\bR\|_\infty \geq \epsilon)\nonumber \\
& \leq \Pr(\bL \geq -\epsilon, \|\bR\|_\infty \leq \epsilon) + \Pr(\|\bR\|_\infty \geq \epsilon) \leq \Pr(\bL \geq -\epsilon) + \Pr(\|\bR\|_\infty \geq \epsilon).
\end{align*}
Then using the decomposition that $\bL = \bL^{(1)} + \bL^{(2)}$, it follows from the same derivation as the above inequality that for any $\epsilon > 0$,
\begin{align*}
\Pr(\brh \geq 0) & \leq \Pr(\bL \geq -\epsilon) + \Pr(\|\bR\|_\infty \geq \epsilon) \\
&\leq \Pr(\bL^{(1)} \geq -2\epsilon) + \Pr(\|\bL^{(2)}\|_\infty \geq \epsilon) + \Pr(\|\bR\|_\infty \geq \epsilon).
\end{align*}
Then by choosing $\epsilon = \frac{1}{2 \sqrt{M}}$, it follows directly from Lemmas~\ref{lem:residual} and~\ref{lem:overlap-cov} that there exist positive constants $C_1, C_2$ and $C_3$ such that
\begin{align} \label{eq:fn}
\Pr(\brh \geq 0) \leq & \Pr(\bL^{(1)} \geq - \frac{1}{\sqrt{M}}) + 2 (d + 2)^2 e^{\frac{1 - \gamma}{2} \log N - C_1 \sqrt{M}} \\
& + 2 (d + 2)^2 e^{\frac{1 - \gamma}{2} \log N - C_2 \sqrt{M}} + 2 (d + 2)^2 e^{\frac{1 - \gamma}{2} \log N - C_3 \frac{M}{T}}. \nonumber
\end{align}
Using that the subsampling assignment is from the random event $\mathcal{B}$, it follows that $T \leq \frac{2 M^2}{N} \cdot N^{\frac{1 - \gamma}{2}}$. By combining this with~\eqref{eq:fn} and the fact that $M = N^\gamma$, we obtain
\begin{align} \label{eq:fn2}
\Pr(\brh \geq 0) & \leq \Pr(\bL^{(1)} \geq - \frac{1}{\sqrt{M}}) + e^{\log (2 (d + 2)^2) + \frac{1 - \gamma}{2} \log N - C_1 N^{\gamma / 2}} \\
& \quad\quad + e^{\log (2 (d + 2)^2) + \frac{1 - \gamma}{2} \log N - C_2 N^{\gamma / 2}} + e^{\log (2 (d + 2)^2) + \frac{1 - \gamma}{2} \log N - C_3 N^{\frac{1 - \gamma}{2}}}. \nonumber
\end{align}
Then using $\log N = o( N^{\gamma / 2 \wedge \frac{1 - \gamma}{2}})$ and $\log (2 (d + 2)^2) = o(N^{\gamma / 2 \wedge \frac{1 - \gamma}{2}})$, we can absorb the terms $\log (2 (d + 2)^2)$ and $\frac{1 - \gamma}{2} \log N$ into $ N^{\gamma / 2}$ and $N^{\frac{1 - \gamma}{2}}$ respectively and obtain
\begin{align*}
\Pr(\brh \geq 0) \leq \Pr(\bL^{(1)} \geq - \frac{1}{\sqrt{M}}) + e^{-C_1 N^{\gamma / 2}} + e^{-C_2 N^{\frac{1 - \gamma}{2}}}. 
\end{align*}
It remains to bound the term $\Pr(\bL^{(1)} \geq -\frac{1}{\sqrt{M}})$. Since all the $\hat{\sigma}_k^{(1)}$'s are independent random vectors, we have
\begin{align*}
\Pr(\bL^{(1)} \geq -\frac{1}{\sqrt{M}}) & = \prod_{k=1}^K \Pr(\frac{M - T_k}{M}\ell_k (\hat{\sigma}_k^{(1)} - \sigma_k) \geq -\frac{1}{\sqrt{M}}) \\
& \leq \prod_{k=1}^K \Pr(\ell_k (\hat{\sigma}_k^{(1)} - \sigma_k) \geq -\frac{2}{\sqrt{M}}),
\end{align*}
where the last inequality is based on the fact that $T_k \ll M$ on the event $\mathcal{B}$ and therefore $\frac{M - T_k}{M} \geq \frac{1}{2}$. Let $\nu_k := M \cdot \textrm{var}(\ell_k (\hat{\sigma}_k^{(1)} - \sigma_k))$. By further applying the standard Berry-Essen theorem, we obtain
\begin{align*}
|\Pr(\ell_k (\hat{\sigma}_k^{(1)} - \sigma_k) \geq - \frac{2}{\sqrt{M}}) - \Pr(Z \geq - 2 / \sqrt{\nu_k}) | \leq C_5 / \sqrt{M},
\end{align*}
where $Z$ represents a standard Gaussian random variable. Using that $\ell_k (\hat{\sigma}_k^{(1)} - \sigma_k)$ can be expressed as the mean of $M - T_k$ independent random variables and that $T_k \ll M$, we obtain that there exists some positive constant $C_4$ such that for all $k \in [K]$, $\nu_k \geq C_4$. Hence, $\Pr(Z \geq - 2 / \sqrt{\nu_k}) \leq \Pr(Z \geq - 2 / \sqrt{C_4})$ and
\begin{align*}
\Pr(\ell_k (\hat{\sigma}_k^{(1)} - \sigma_k) \geq - \frac{2}{\sqrt{M}}) \leq \Pr(Z \geq - 2 / \sqrt{C_4}) + C_5 / \sqrt{M} \leq C_6
\end{align*}
for some positive constant $C_6 < 1$. Hence, one can rewrite~\eqref{eq:fn2} as
\begin{align*}
\Pr(\brh \geq 0) \leq (C_6)^K + e^{-C_1 N^{\gamma / 2}} + e^{-C_2 N^{\frac{1 - \gamma}{2}}},
\end{align*}
which finally yields
\begin{align*}
\Pr(\brh \geq 0) \leq e^{- (\log \frac{1}{C_6}) \cdot N^{\frac{1 - \gamma}{2}}} + e^{-C_1 N^{\gamma / 2}} + e^{-C_2 N^{\frac{1 - \gamma}{2}}}
\end{align*}
under the random event $\mathcal{B}$, which completes the proof.
\end{proof}

\section{Proof of Theorem~\ref{thm:fp}}

\begin{proof}[Proof of Theorem~\ref{thm:fp}]
For any $i \neq j$, without loss of generality, we assume that $|\adj_i(G)| \leq |\adj_j(G)|$. Also, let $S_{ij} := \adj_i(G) \setminus \{j\}$. We denote the random event $\mathcal{A}$ by:
\begin{align*}
\mathcal{A} := \Big\{\textrm{for any} \;(i,j) \not\in G, \exists t \in [p] \setminus S_{ij} \cup \{i,j\} \;\textrm{such that}\; \hat{\rho}_{i,j \mid S_{ij} \cup \{t\}} \leq 0\Big\}.
\end{align*}
Similarly, for each $(i,j) \not\in G$, we let
\begin{align*}
\mathcal{A}_{ij} := \Big\{\exists t \in [p] \setminus S_{ij} \cup \{i,j\} \;\textrm{such that}\; \hat{\rho}_{i,j \mid S_{ij} \cup \{t\}} \leq 0\Big\}.
\end{align*}
Let $t_1, \cdots, t_K \in [p] \setminus S_{ij} \cup \{i,j\}$ denote a list of nodes with size $K = N^{\frac{1 - \gamma}{2}}$ (this is a valid choice since Condition~\ref{cd:size} gives us that $p \geq  N^{\frac{1 - \gamma}{2}} + d + 2$ for any $\gamma \in (\frac{3}{4}, 1)$). It is straightforward to show that $\rho_{ij \mid S_{ij} \cup \{t_k\}} = 0$ for all $k \in [K]$. Then by setting each $S_k$ in Lemma~\ref{lem:fn} as $S_k := S_{ij} \cup \{t_k\}$, it follows from Lemma~\ref{lem:fn} that with probability at least $1 - \exp(-CN^{\frac{1 - \gamma}{2} \wedge 4\gamma - 3})$, there exists some $t_k$ such that $\hat{\rho}_{i,j \mid S_{ij} \cup \{t_k\}} \leq 0$, which yields $\Pr(\mathcal{A}_{ij}) \geq 1 - \exp(-CN^{\frac{1 - \gamma}{2} \wedge 4\gamma - 3})$. By taking the union bound over all the edges $(i,j) \not\in G$, we obtain that $\Pr(\mathcal{A}) \geq 1 - p^2e^{-C\frac{1 - \gamma}{2} \wedge 4 \gamma - 3}$.

Thus, to complete the proof of the theorem, it remains to prove that under the random event $\mathcal{A}$, all edges $(i,j) \not\in G$ are deleted by Algorithm~\ref{alg:mtp2} when the algorithm is at iteration $\ell = d + 1$. We prove this by contradiction. Suppose there exists an edge $(i,j) \not\in G$ that is not deleted by the algorithm at $\ell = d + 1$. By applying Theorem~\ref{thm:fn}, we obtain that the estimated graph $\hat{G}$ in the iteration $\ell = |\adj_i(G)|$ satisfies $\adj_i(G) \subseteq \adj_i(\hat{G})$ and as a consequence the edge $(i,j)$ will be selected at Step~5 of Algorithm~\ref{alg:mtp2} at iteration $\ell = |\adj_i(G)|$. Then by choosing the $S$ at Step~7 to be $S_{ij}$ and using that we are on the event $\mathcal{A}$, we obtain that there exists a node $k$ such that $\hat{\rho}_{ij \mid S \cup \{k\}} \leq 0$. As a consequence, the edge $(i,j)$ will be deleted at Step~8. This contradicts the fact that the edge $(i,j)$ exists in the final output, which completes the proof.
\end{proof}

\section{Proof of Theorem~\ref{thm:fn}}

\begin{lemma}\label{lem:rho}
Consider a Gaussian random vector $X = (X_1, \cdots, X_p)^T$ that follows an MTP$_2$ distribution. Then for any $i,j \in [p]$ and any $S \subseteq [p] \setminus \{i,j\}$, it holds that $\rho_{ij \mid S} \geq \rho_{ij \mid [p] \setminus\{i,j\}}$.
\end{lemma}

\begin{proof}
For $\rho_{ij \mid S}$, if we let $M = S_{i,j}$, we have
\begin{align*}
    \rho_{ij \mid S} = - \frac{((\Sigma_M)^{-1})_{i_M, j_M}}{\sqrt{((\Sigma_M)^{-1})_{i_M, i_M} ((\Sigma_M)^{-1})_{j_M, j_M}}}.
\end{align*}
Using that the precision matrix $\Theta$ is an M-matrix, it follows from basic calculations using Schur complements that $((\Sigma_M)^{-1})_{i_M, i_M} \leq \Theta_{ii}$, $((\Sigma_M)^{-1})_{j_M, j_M} \leq \Theta_{jj}$ and $((\Sigma_M)^{-1})_{i_M, j_M} \leq \Theta_{ij} \leq 0$. By combining this with the fact that $\rho_{ij \mid [p] \setminus \{i,j\}} = -\frac{\Theta_{ij}}{\sqrt{\Theta_{ii}\Theta_{jj}}}$, we obtain the lemma.
\end{proof}

With this, we can now provide the proof of Theorem~\ref{thm:fn}.

\begin{proof}
For any edge $(i,j) \in G$ and any conditioning set $S \subseteq [p] \setminus \{i,j\}$ with $|S| \leq d + 2$, by using the same decomposition as in~\eqref{eq:fn2}, we can decompose the random variable $\hat{\rho}_{ij \mid S}$ as
\begin{align*}
\hat{\rho}_{ij \mid S} = \rho_{ij \mid S} + \ell_{ij \mid S} + r_{ij \mid S},
\end{align*}
where the random variable $\ell_{ij \mid S}$ is the first order approximation of $\hat{\rho}_{ij \mid S} - \rho_{ij \mid S}$ and $r_{ij \mid S}$ is the residual. It follows from Lemma~\ref{lem:concentration} and the proof of Lemma~\ref{lem:residual} that there exists some positive constant $\tau$ such that with probability at least $1 - p^{-(\tau + d + 4)}$,
\begin{align*}
|\hat\rho_{ij \mid S} - \rho_{ij \mid S}| \leq C_1 \sqrt{ (\tau + d + 4)\frac{\log p}{N^\gamma}},
\end{align*}
where $C_1$ is some positive constant that depends on $\sigma_{\min}, \sigma_{\max}$ and $d$. By further taking union bound over all $(i,j) \in G$ and all $S \subseteq [p] \setminus \{i,j\}$ with $|S| \leq d + 2$, it follows that
\begin{align*}
\Pr\Bigg\{\forall (i,j) \in G, \; \forall S \subseteq [p] \setminus \{i,j\} \;\textrm{with}\; |S| \leq d + 2, |\hat\rho_{ij \mid S} - \rho_{ij \mid S}| \leq C_1 \sqrt{(\tau + d + 4) \frac{\log p}{N^\gamma}}\Bigg\} \\
\geq 1 - p^{-\tau}.
\end{align*}
As a consequence, by assuming that $c_\rho$ in Condition~\ref{cd:signal} is sufficiently large such that $c_\rho > C_1 \sqrt{d + 4}$ and choosing $\tau$ such that $\tau < \big(\frac{c_\rho}{C_1}\big)^2 - d - 4$, it follows from Lemma~\ref{lem:rho} that with probability at least $1 - p^{-\tau}$, $\hat\rho_{ij \mid S} > 0$ for all the $(i,j,S)$'s where $(i,j) \in G$ and $|S|\leq d + 2$. Hence, we obtain that the edges $(i,j) \in G$ will not be deleted by Algorithm~\ref{alg:mtp2}, which completes the proof.
\end{proof}

\section{Additional comments on empirical evaluation}

\subsection{Stability selection}

\paragraph{Overview of stability selection: } Stability selection~\citep{MB06} is a well-known technique for enhancing existing variable selection algorithms with tuning parameters. Stability selection works by taking an existing algorithm with a tuning parameter and running it multiple times on different subsamples of the data with various reasonable values for the tuning parameter. A variable is selected if there exists a tuning parameter for which it is selected often enough (in our case we use the threshold $\pi = 0.8$, meaning a variable must be present in at least $80\%$ of trials for a given tuning parameter). Because for each tuning parameter, the algorithm is run many times on different subsamples of the data, stability selection is very computationally expensive. It is important to note that stability selection is \emph{better} than simply choosing the best tuning parameter for a given algorithm, as it is able to combine information across various tuning parameters where appropriate and adapt to different settings. 

\paragraph{The advantages of stability selection: } As can be seen from Figure~\ref{fig:mcc}(c), the purple line corresponds to the \emph{SH} algorithm with stability selection and the pink line corresponds to the \emph{SH} algorithm where the \emph{best} tuning parameter is chosen for each different $N$ (i.e. the $y$-axis contains the \emph{best} MCC across \emph{all} tuning parameters). Note that the pink line is not a realistic scenario, as in a real-world application we would not have access to the evaluation metric on the test dataset as we do in this simulated example. However this example is instructive in showing that \emph{even when} a particular algorithm is evaluated with the best possible tuning parameter, stability selection is able to outperform it, showing that stability selection truly offers a tremendous advantage for the performance of algorithms with tuning parameters. Thus it is remarkable that our algorithm with theoretically optimal $\gamma$ is able to compete with other algorithms using stability selection.

\paragraph{Variation of $\gamma$ and our algorithm with stability selection:} It is also worth noting that although our algorithm doesn't have a ``tuning parameter" in a traditional sense (i.e. our consistency guarantees are valid for all $\gamma \in (0.75,1)$), it is still possible to perform stability selection with our algorithm by using various choices of $\gamma$ in the valid range. In particular, we see from Figure~\ref{fig:mcc}(c) that our algorithm with $\gamma = 0.85$ out-performs the theoretically ``optimal" value of $\gamma = 7/9$. Thus in practice, because different values of $\gamma$ lead to different performance (and in some cases better performance than the theoretically optimal value), our algorithm would likely be improved by performing stability selection. This would likely offer an improvement in performance for our algorithm at the expense of higher computational costs. Although it is worth noting that in our experiments our algorithm \emph{without} stability selection performed quite competitively.

\subsection{FPR and TPR}

In Figures~\ref{fig:tpr} and~\ref{fig:fpr} we report performance of various methods based on the false positive rate (FPR) and true positive rate (TPR) respectively. From these figures we can get similar conclusion as using the MCC measure. In particular, it is important to note that although the TPR of CMIT is higher than our algorithm across all simulation set ups, its FPR is also high, which makes the overall performance less compelling than our algorithm. The performance of TIGER is worse than our method in terms of both TPR and FPR.

\begin{figure*}[!t]
	\centering
	\subcaptionbox{Random graphs}{\includegraphics[width=0.33\textwidth]{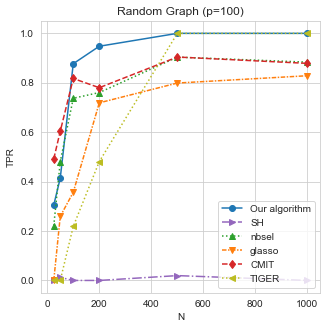}}%
	\subcaptionbox{Chain graphs}{\includegraphics[width=0.33\textwidth]{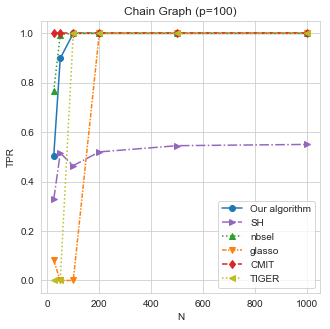}}%
	\subcaptionbox{Grid graphs}{\includegraphics[width=0.33\textwidth]{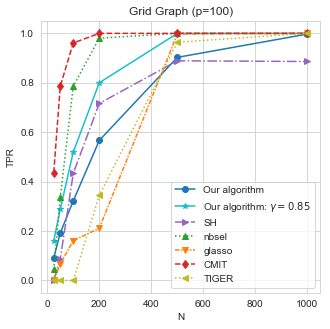}}%
	\caption{Comparison of different algorithms evaluated on TPR.}
		\label{fig:tpr}
\end{figure*}

\begin{figure*}[!t]
	\centering
	\subcaptionbox{Random graphs}{\includegraphics[width=0.33\textwidth]{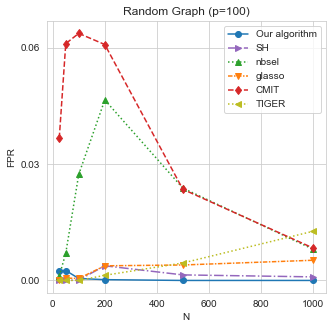}}%
	\subcaptionbox{Chain graphs}{\includegraphics[width=0.33\textwidth]{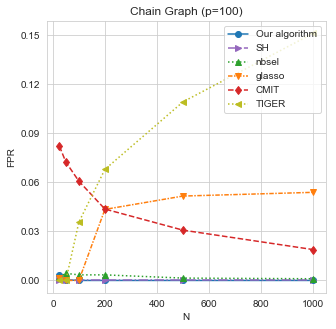}}%
	\subcaptionbox{Grid graphs}{\includegraphics[width=0.33\textwidth]{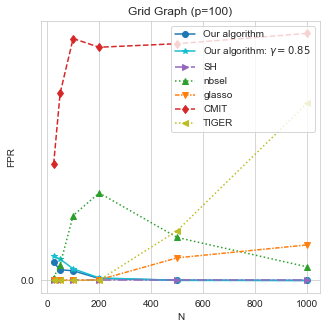}}%
	\caption{Comparison of different algorithms evaluated on FPR.}
		\label{fig:fpr}
\end{figure*}

\subsection{ROC Curves}

To generate the ROC curve for each setting of $N$, we sample $30$ different random graphs (\emph{random} as defined in Section~\ref{sec:eval})) and then draw $N$ samples from a multivariate normal with the resulting precision matrix. For each of the $30$ trials, we get an ROC curve for each algorithm based on the range of tuning parameters tried. To get a mean ROC curve for each algorithm, we average together the $30$ trials. The averaged ROC curves are shown Figure~\ref{fig:roc}(a) as well as Figure~\ref{fig:appendix_ROC}. The range of tuning parameters tried for each algorithm is listed below: 
\begin{itemize}
    \item \emph{SH: } $20$ equally spaced points for $q \in [0.00, 1.0]$. 
    \item \emph{glasso, nbsel: } $20$ equally spaced points in $\log$ space for $\lambda \in [10^{-6},10^{1.2}]$.
    \item \emph{CMIT: } For computational reasons, we always set $\eta=1$. However the tuning threshold $\lambda$ is varied as 20 equally spaced points in log space between $\lambda \in [10^{-4}, 10^{1.2}] $.
    \item \emph{Our algorithm: } We varied $\gamma \in [0.75, 0.95]$ for 10 equally spaced points in this interval.
\end{itemize}

\begin{figure}[!htb]
	\centering
	\subcaptionbox{ROC Curve $N=50$}{\includegraphics[width=0.33\textwidth]{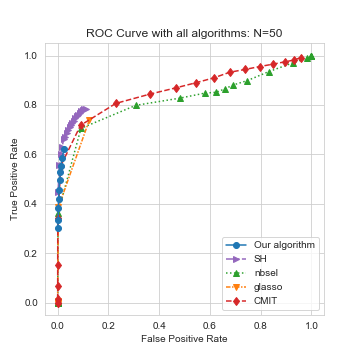}}%
	\subcaptionbox{ROC Curve $N=100$}{\includegraphics[width=0.33\textwidth]{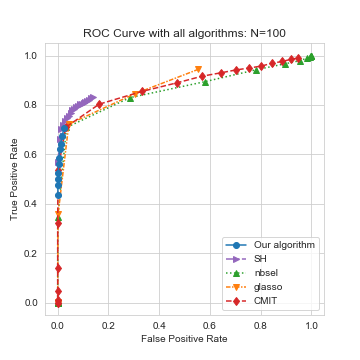}}%
	\subcaptionbox{ROC Curve $N=200$}{\includegraphics[width=0.33\textwidth]{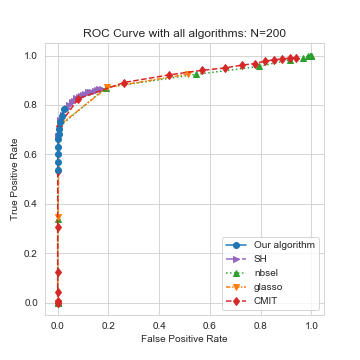}}%
	\caption{ROC curves for $N=50, 100, 200$ respectively averaged across $30$ trials of a random graph with $p=100$.}
	\label{fig:appendix_ROC}
\end{figure}

\subsection{Normalization of Tuning Parameters}

For each algorithm there is a \emph{reasonable} range of tuning parameters that one might consider while attempting to perform structure recovery for Gaussian graphical models with the particular algorithm \emph{in practice}. For \emph{glasso} and \emph{nbsel} it is well known that $\lambda = O\left( \sqrt{\frac{\log p}{N}} \right)$ is theoretically optimal~\citep{FHT08, MB06}. For all of the experiments shown in Figure~\ref{fig:roc}, we have that $p=100$ and $N=500$, giving $\sqrt{\frac{\log p}{N}} \approx 0.1$. To test the sensitivity of these algorithms' performance to choice of $\lambda$ close to this optimal quantity, we let the minimum and maximum $\lambda$ for both of these algorithms be a factor of $5$ within $0.1$. Thus, $\lambda_{\min(\text{glasso})} =\lambda_{\min(\text{nbsel} )} = 0.02$ and $\lambda_{\max(\text{glasso})} =\lambda_{\max(\text{nbsel} )} = 0.5$. We ran both algorithms with a variety of tuning parameters in this range and mapped the tuning parameters linearly to $[0, 1]$ so that $0.02$ is mapped to $0$ and $0.5$ is mapped to $1$ in the normalized tuning parameter $x$-axis in Figures~\ref{fig:roc}(b) and (c). For \emph{CMIT}, the threshold is also optimal for $ O\left( \sqrt{\frac{\log p}{N}} \right)$, so we chose $\eta = 1$ for computational reasons and let the threshold vary similarly as \emph{glasso} and \emph{nbsel} and be mapped to $[0, 1]$ similarly for normalization. 

For \emph{SH}, we let the threshold $q \in [0.7, 1.]$ as that is the range of threshold quantiles that the authors used in their paper~\citep{SH15}. 
Once again, we performed a linear transformation such that the interval of tuning parameters gets mapped to the unit interval. 

For our algorithm, we let $\gamma \in [0.75, 0.95]$ and also mapped this interval to $[0,1]$ for normalizing the $\gamma$ ``tuning parameter". We decided this was an appropriate range for $\gamma$ since the Algorithm is consistent for $\gamma \in (0.75, 1)$. We make a minor note that in our mapping, we let smaller values of $\gamma$ correspond to higher values of the normalized tuning parameter (still a linear mapping, simply a reflection of the $x$-axis) since as $\gamma$ decreases, it performs similarly to providing more regularization since more edges are removed. In general, an increase in the normalized tuning parameter corresponds to more regularization.

Throughout, we wanted to use a reasonable range of tuning parameters for all algorithms to map onto the unit interval after normalization, so that we could have a fair comparison of the sensitivity of different algorithms' performance to their respective choice of tuning parameters. 

\section{Real data analysis}

In this analysis, we consider the following metric that evaluates the community structure of a graph.

\noindent \emph{Modularity.} Given an estimated graph $G := ([p], E)$ with vertex set $[p]$ and edge set $E$, let $A$ denote the adjacency matrix of $G$. For each stock $j$ let $c_j$ denote the sector to which stock $j$ belongs and let $k_j$ denote the number of neighbors of stock $j$ in $G$. Then the \emph{modularity coefficient} $Q$ is given by%
$$ Q = \frac{1}{2 |E|} \sum_{i, j \in [p]} \left( A_{ij} - \frac{k_ik_j}{2|E|} \right) \delta(c_i, c_j),$$%
where $\delta(\cdot,\cdot)$ denotes the $\delta$-function with $\delta(i, j) = 1$ if $i = j$ and 0 otherwise. 

The modularity coefficient measures the difference between the fraction of edges in the estimated graph that are within a sector as compared to the fraction that would be expected from a random graph. A high coefficient $Q$ means that stocks from the same sector are more likely to be grouped together in the estimated graph, while a low $Q$ means that the community structure of the estimated graph does not deviate significantly from that of a random graph. Table~1 in the main paper shows the modularity scores of the graphs estimated from the various methods; our method using fixed $\gamma = 7/9$ outperforms all the other methods. 


\end{document}